\documentclass{article}
\usepackage{graphicx,psfrag,authblk}
\usepackage[small,isu]{caption}
\usepackage{latexsym}
\usepackage{amsmath,amsthm,bbm,mathrsfs}
\usepackage{pdfsync}
\usepackage{amssymb}
\usepackage[italian,francais,german,english]{babel}
\def\d{{\mathrm d}}

\def\ZZ{{\mathbb Z}}

\newcommand{\beq}{\begin{equation}}
\newcommand{\norm}[1]{\lVert #1 \rVert}
\newcommand{\eeq}{\end{equation}}
\newcommand{\beqn}{\begin{eqnarray}}
\newcommand{\eeqn}{\end{eqnarray}}
\newcommand{\bal}{\begin{align}}
\newcommand{\eal}{\end{align}}
\newcommand{\scalar}[2]{\langle{#1} \mspace{2mu}, {#2}\rangle}
\def\CC{{\mathbb C}}
\def\({\left(}
\def\){\right)}
\def\L{\varLambda}
\def\l{\lambda}
\def\EE{{\mathbbmss E\,}}
\def\11{{\mathbbmss 1}}
\def\PP{{\mathbbmss P}}
\def\i{{\mathrm i}}
\def\e{{\mathrm e}}
\def\tr{{\mathrm{tr}\;}}
 
\def\eps{{\varepsilon}}

\newtheorem{thm}{Theorem}
\newtheorem*{lem*}{Lemma}

\begin{document}
\title{Anderson Localization Triggered by Spin Disorder---with an Application to $\mathrm{Eu}_x\mathrm{Ca}_{1-x}\mathrm{B}_6$}
\author[1]{Daniel Egli\thanks{danegli@phys.ethz.ch}}
\author[1]{J\"urg Fr\"ohlich\thanks{juerg@phys.ethz.ch}}
\author[3]{Hans-Rudolf Ott\thanks{ott@phys.ethz.ch}}
\affil[1]{Institute for Theoretical Physics \\ ETH Z\"urich \\ Wolfgang-Pauli-Strasse 27 \\ CH-8093 Z\"urich}
\affil[2]{Laboratory for Solid State Physics\\ ETH Z\"urich\\ Schafmattstrasse\\ CH-8093 Z\"urich}
\maketitle

\abstract{The phenomenon of Anderson localization is studied for a class of one-particle Schr\"odinger operators with random Zeeman interactions. These operators arise as follows: Static spins are placed randomly on the sites of a simple cubic lattice according to a site percolation process with density $x$ and coupled to one another ferromagnetically. Scattering of an electron in a conduction band at these spins is described by a random Zeeman interaction term that originates from indirect exchange. It is shown rigorously that, for positive values of $x$ below the percolation threshold, the spectrum of the one-electron Schr\"odinger operator near the band edges is dense pure-point, and the corresponding eigenfunctions are exponentially localized.

Localization near the band edges persists in a weak external magnetic field, $H$, but disappears gradually, as $H$ is increased. Our results lead us to predict the phenomenon of colossal (negative) magnetoresistance and the existence of a Mott transition, as $H$ and/or $x$ are increased.

Our analysis is motivated directly by experimental results concerning the magnetic alloy $\mathrm{Eu}_x\mathrm{Ca}_{1-x}\mathrm{B}_6$.}

\section{Introduction}
Theoretical understanding of the effects of disorder on electron transport in metals and semiconductors started with Anderson's famous 1958 paper \cite{anderson58}. In this work, Anderson argued that in the presence of strong disorder caused by impurities and/or defects and neglecting electron-electron interactions, electrons populating a weakly filled conduction band of a metal get trapped in exponentially sharply localized one-particle orbitals. This is the phenomenon known as ``Anderson localization''. A consequence of localization is that the conductivity of such a material very nearly vanishes at low temperatures. If disorder is described by an on-site random potential with bounded distribution and short-range correlations in a one-electron tight-binding Hamiltonian then Anderson's arguments can be made precise, mathematically, for one-dimensional systems with arbitrarily weak disorder \cite{Goldsheid77,Kunz80}, and for higher-dimensional systems, provided that disorder is strong enough or the energy of the one-electron orbital lies in the band tails \cite{froehlich83,froehlich85}. It is generally expected---but not rigorously proven---that, in two-dimensional systems of this kind, all states are localized, no matter how weak the disorder is. In contrast, in three or more dimensions, localized states with energies in the band tails are expected to coexist with extended states (generalized eigenstates of the model Hamiltonian) corresponding to energies in the continuous spectrum near the center of the band, provided the disorder is sufficiently weak. It is expected that wave packets made from superpositions of such extended states exhibit diffusive propagation corresponding to a non-zero conductivity \cite{wegner79,abrahams79}. One is led to predict that, at very low temperatures, a three-dimensional disordered semiconductor exhibits a transition from an insulating state (all electrons in the conduction band occupy localized states) to a conducting state (some fraction of the electrons populate extended states), as the density of electrons in the conduction band is enhanced or the strength of disorder is lowered. This transition from an insulator to a metal is called a ``Mott transition''.

The purpose of this paper is to describe and analyze a tight-binding model for a Mott transition where the disorder is caused by indirect exchange interactions between the electrons in a conduction band and a dilute array of localized atoms with non-vanishing spin. This model has been introduced in order to interpret theoretically electronic properties of Europium-based hexaborides ($\mathrm{Eu}_x\mathrm{Ca}_{1-x}\mathrm{B}_6$), \cite{beeli04}. A Mott transition is found experimentally as the concentration, $x$, of the magnetic Europium atom is varied. This feature is captured correctly in our model. Furthermore, the model also suggests an explanation of the phenomenon of ``colossal negative magneto-resistance'' encountered in experimental explorations of $\mathrm{Eu}_x\mathrm{Ca}_{1-x}\mathrm{B}_6$; see \cite{beeli04}.

Fairly simple arguments presented below lead us to introduce a model given in terms of a one-electron tight-binding Hamiltonian with a \textbf{random Zeeman interaction term} acting on electron spin. This term describes indirect exchange interactions between an electron in the conduction band and electrons in the half-filled $4f$ shell of a Eu-atom located nearby. It takes the form of a ferromagnetic coupling of the spin of the electron in the conduction band to the static total spin of electrons in the $4f$ shell of a Eu-atom. Because the latter is quite large, $S=7/2$, it can be described, in good approximation, by a classical unit vector, $\vec{m}$ \cite{lieb73}. However, if a unit cell of the simple cubic lattice of a $\mathrm{Eu}_x\mathrm{Ca}_{1-x}\mathrm{B}_6$ alloy contains a Ca-atom then $\vec{m}=0$, because a Ca-atom has spin 0. At low temperatures, the direction of $\vec{m}$ is approximately constant throughout a connected Eu-cluster, because of indirect ferromagnetic exchange interactions between the spins of different Eu-atoms in the cluster. However, the direction of $\vec{m}$ varies randomly from one Eu-cluster to another, as long as the external magnetic field vanishes. Thus, an electron in the conduction band of a $\mathrm{Eu}_x\mathrm{Ca}_{1-x}\mathrm{B}_6$ alloy in zero magnetic field propagates in a disordered quasi-static background of essentially \textbf{classical spins} located in those unit cells that contain a Eu-atom. These spins are ferromagnetically coupled to the spin operator of the electron.

One of our main contentions in this paper is that, as long as there is no ferromagnetic long-range order (unit cells containing a Eu-atom do not percolate), but the concentration of Eu-atoms is not too small, in zero magnetic field, this type of magnetic disorder causes \textbf{Anderson localization} in the tails of the conduction band.

If the concentration, $x$, of Eu-atoms is brought above the percolation threshold then there is an infinite connected cluster of positive density of unit cells containing a Eu-atom, and the alloy is observed to order ferromagnetically at low enough temperatures \cite{wigger05}. Most Eu-spins are then aligned in a fixed direction. The same happens if a sufficiently strong external magnetic field is applied. Finally, if $x$ is very small most unit cells exhibit a vanishing spin, that is, the vector $\vec{m}$ vanishes in most unit cells. In all these three situations, the disorder felt by electrons in the conduction band is weak, so that the localization threshold (or ``mobility edge'') moves towards the band edges. We thus expect to observe a delocalization-- or Mott transition to a conducting state, as $x$ increases across $x_{\rm c}$, or if the external magnetic field is increased.

It is not understood, at present, how to \textbf{prove} the existence of such a transition and analyze its characteristics, although, heuristically, it is fairly well understood. But the existence of Anderson localization in the band tails, for $x$ below the percolation threshold, but non-zero, and for a sufficiently weak external magnetic field, can be proven rigorously. A sketch of our proof is the main message of our paper; technical details can be found in \cite{doktorarbeit}. We have profited from previous results in \cite{bourgain09,spencer93}.

While the main mathematical results presented in this paper may not be particulary suprising, they concern examples of Anderson localization that have \textbf{not} previously been studied mathematically.

Our discussion is summarized Figures 1 and 2.

\begin{figure}[h]\center
\psfrag{metallic}{metallic behavior}
\psfrag{insulator}{insulator}
\psfrag{1}{$1$}
\psfrag{0}{$0$}
\psfrag{H}{$H$}
\psfrag{x}{$x$}
\psfrag{xc}{$x_{\rm c}$}
\includegraphics[scale=0.5]{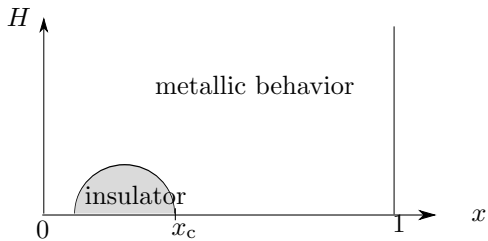}
\caption{$x$ denotes the concentration of Eu-atoms, $H$ the value of a homogeneous external magnetic field. The concentration of conduction electrons is assumed to be approximately constant. The shaded area corresponds to an insulating state; a Mott transition to a semi-metal is expected to be observed at its boundary. Rigorous results are known for a subset of the parameter values inside the shaded area.}\label{figure1}
\end{figure}

\begin{figure}[h]\center
\psfrag{H}{$H$}
\psfrag{sigma}{$\sigma$}
\psfrag{0}{$0$}
\psfrag{1}{$1$}
\psfrag{x}{$x$}
\includegraphics[scale=0.5]{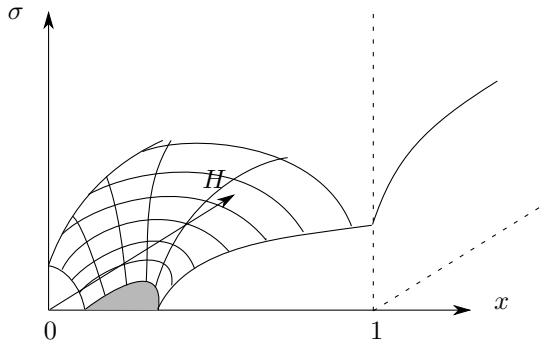}
\caption{$\sigma$ denotes the conductivity of the alloy. The figure provides a \textbf{qualitative} plot of $\sigma$.}\label{figcond}
\end{figure}

Our paper is organized as follows. In Section 2, we briefly recapitulate some experimental findings on $\mathrm{Eu}_x\mathrm{Ca}_{1-x}\mathrm{B}_6$ and sketch physical mechanisms that may be at the root of the observed phenomena. They lead us to propose, in Section 3, some idealized models that are expected to yield an adequate theoretical interpretation of the observed phenomena. In Section 4, we formulate our main mathematical results on our models and sketch how they may be used to interpret various features found in the experiments. In Section 5, brief outlines of proofs of our main results are presented; (for details, the reader is referred to \cite{doktorarbeit}).

\section{Summary of experimental results concerning $\mathrm{Eu}_x\mathrm{Ca}_{1-x}\mathrm{B}_6$, and physical mechanisms}
We begin by recalling some essential properties of $\mathrm{Eu}\mathrm{B}_6$. This binary compound crystallizes in a simple cubic lattice. At the center of each unit cell of the crystal there is a divalent Eu atom, at every corner of a unit cell there is an octahedron of B-ions; see Figure 3 below.

\begin{figure}[h]\center
\psfrag{B}{B}
\psfrag{Eu}{Eu}
\includegraphics[scale=0.5]{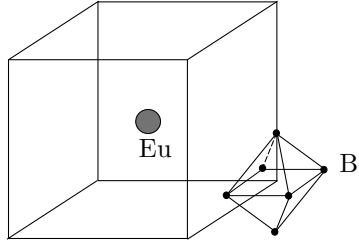}
\caption{Schematic unit cell of $\mathrm{Eu}\mathrm{B}_6$; each cube corner is the centre of a Boron octahedron}\label{}
\end{figure}
The $4f$-shell of a Eu-atom is half filled, which, according to Hund's rule, implies that the total spin is $s=7/2$. Electron transport is dominated by defect-state conduction, with a low concentration, $n_{c}$, of around $10^{-3}$ charge carriers per unit cell \cite{wigger04}. At low temperatures, $\mathrm{Eu}\mathrm{B}_6$ orders ferromagnetically at a Curie temperature $T_{C}\simeq$ 12 K, accompanied by a significant reduction of the resistivity, $\rho$, in the ordered phase \cite{fisk79}. The isostructural compound $\mathrm{Ca}\mathrm{B}_6$ is obtained by replacing Eu by isoelectronic but non-magnetic Ca, which leads to a further reduction of $n_{\rm c}$ by an order of magnitude \cite{beeli04}. In the series $\mathrm{Eu}_x\mathrm{Ca}_{1-x}\mathrm{B}_6$, $T_{\rm C}$ decreases monotonically with decreasing $x$, down to $x\simeq 0.3$. At lower values of $x$, no onset of long-range magnetic order is observed. Instead spin-glass type features dominate the magnetic response at low temperatures \cite{wigger05}.  For the simple cubic Eu sublattice, $x_{\rm c} = 0.31$ is the site percolation limit \cite{wigger05}. In the concentration range $0.2<x<0.3$, significant localization and colossal magnetoresistance effects, such as shown in Figure 4, have been observed. For $x = 0.27$, the enhancement of the low-temperature resistivity by six orders of magnitude below 10 K may be quenched by rather moderate magnetic fields of the order of 1 T.
\begin{figure}[h]\center
\psfrag{B}{B}
\psfrag{Eu}{Eu}
\includegraphics[scale=0.4]{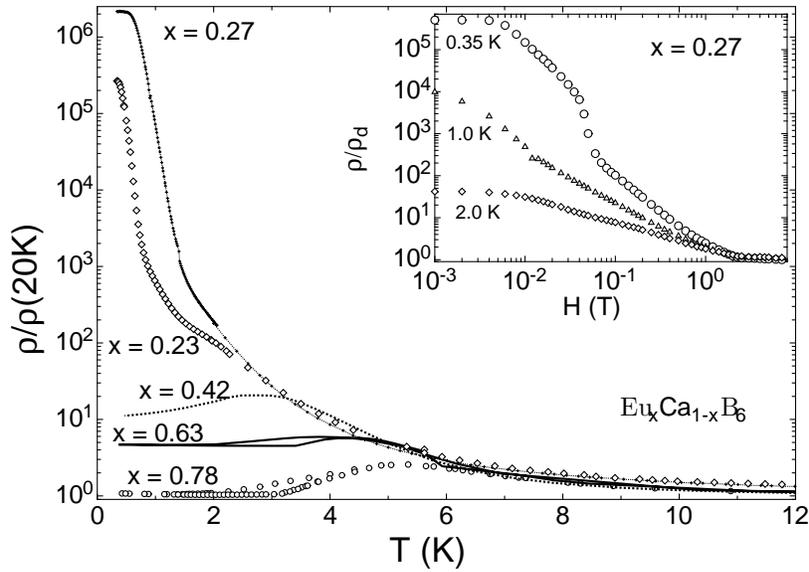}
\caption{Main panel: $\rho(T)$ of $\mathrm{Eu}_x\mathrm{Ca}_{1-x}\mathrm{B}_6$ for various values of $x$. The thin solid line for $x=0.27$ is to guide the eye. Inset: Magnetoresistance of $\mathrm{Eu}_{0.27}\mathrm{Ca}_{0.73}\mathrm{B}_6$ at low temperatures.}\label{}
\end{figure}
Detailed investigations using selected-area electron diffraction patterns and high-resolution transmission electron-microscopy (HRTEM) have shown that also for large concentrations of Ca for Eu, the structural quality, that is, the perfect atomic arrangement in a simple cubic lattice is preserved and the disorder is simply in the spins on the sites of the Eu clusters. Energy-filtered TEM reveals a phase separation into microscopically small Ca- and Eu-rich regions, respectively. This implies that the material is magnetically and electronically inhomogeneous \cite{beeli04}.

Next, we sketch some ideas on a possible mechanism that may explain the long-range ferromagnetic order observed in $\mathrm{Eu}\mathrm{B}_6$ at temperatures below $T_{\rm C}$. See also \cite{pereira04} for a similar discussion. The large size of the unit cells of $\mathrm{Eu}\mathrm{B}_6$ (as compared to the size of a Eu atom) and numerical simulations \cite{kunes04} suggest that ferromagnetic order is established through \textbf{indirect exchange} mediated by electrons in a somewhat less than half-filled valence band, with strong on-site Coulomb repulsion preventing double occupancy; see Figure 5. For a non-vanishing density of holes in the valence band \cite{wigger04}, the spins of the electrons in the valence band are expected to order ferromagnetically at very low temperatures. For the groundstate, this is a prediction of the Thouless-Nagaoka theorem \cite{thouless65,nagaoka66,aizenman90}; (see also \cite{ueltschi05} for an analysis of ferromagnetism in the Hubbard model). Because of overlap of the orbitals of electrons in the valence band with those in the $4f$-shells of Eu-atoms, the spin of a valence electron in a unit cell has a tendency of being ``anti-parallel'' to the total spin of the Eu-atom in the same unit cell, provided the temperature is low. Appealing to Hund's rule, this is seen to be a consequence of Pauli's exclusion principle and of the half-filling of the $4f$-shell. Hopping processes of valence electrons into either an empty orbital of the $4f$-shell of a Eu-atom or to an empty orbital of the valence band thus give rise to ferromagnetic order among the spins of the Eu-atoms and those of the valence electrons, the latter being ``anti-parallel'' to the spins of the Eu-atoms.
\begin{figure}[h]\center 
\psfrag{C}{conduction band}
\psfrag{V}{valence band}
\psfrag{LM}{localized Eu-spins}
\includegraphics[scale=0.5]{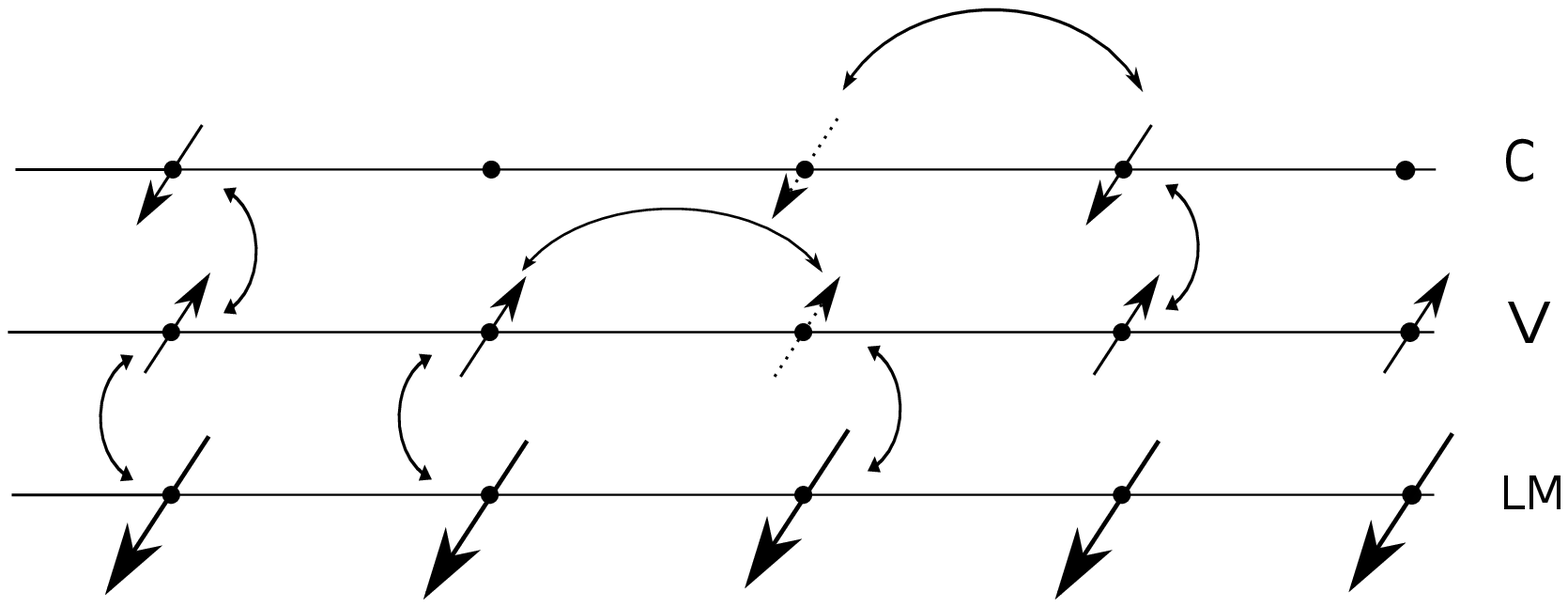}
\caption{Spin ordering in $\mathrm{Eu}\mathrm{B}_6$.}\label{}
\end{figure}
Because the orbitals of conduction electrons overlap with those of valence electrons, there are exchange interactions between conduction-- and valence electrons that, because of the Pauli principle, favor anti-ferromagnetic order between conduction-- and valence electrons. Thus, the spins of conduction electrons have a tendency of being aligned with the spins of the Eu-atoms. We will describe this tendency by a Heisenberg term that couples the spin of a conduction electron in a unit cell ferromagnetically to the spin of the Eu-atom in the same unit cell. Since the spin of the Eu-atom is rather large ($S=7/2$), we propose to describe it as a static classical spin, $\vec{m}$. It would be of considerable interest to improve the theoretical understanding of ferromagnetism in a one-band Hubbard model coupled to a lattice of large localized spins.

Our main concern, in this paper, is to provide some qualitative understanding of electronic properties of $\mathrm{Eu}_x\mathrm{Ca}_{1-x}\mathrm{B}_6$. Experimentally, a Mott transition from a metallic state to an insulator is observed, as the concentration, $x$, of Eu is lowered. Furthermore, for $x\lesssim 0.3$, when a homogeneous external magnetic field is turned on, ``colossal negative magnetoresistance'' is observed, \cite{beeli04}. These are the phenomena we wish to focus on in this paper.

In our somewhat idealized theoretical description of $\mathrm{Eu}_x\mathrm{Ca}_{1-x}\mathrm{B}_6$, we place Eu-- and Ca-atoms at the centers of the unit cells of the simple cubic lattice $\ZZ^3$ according to a site percolation process, with probability $x$ to find a Eu-atom at a given site. The mechanism for ferromagnetic order through indirect exchange described above suggests that, within connected clusters of unit cells filled with Eu-atoms, the spins of the Eu-atoms are ferromagnetically ordered. Since different Eu-clusters are separated by regions filled with non-magnetic Ca-atoms, one expects that the directions in which the spins of Eu-atoms are aligned vary randomly from one Eu-cluster to the next, as long as the external magnetic field vanishes (or is very small). If there is no infinite cluster of Eu-atoms this introduces disorder, and, because the conduction electrons are scattered at the spins of the Eu-atoms, it enhances a tendency towards Anderson localization of conduction electrons.

The threshold for the emergence of an infinite connected cluster in a site percolation process on $\ZZ^3$ is $x_{\rm c}\simeq 0.31$. For $x$ above $x_{\rm c}$, one expects that there exists an infinite connected cluster of Eu-atoms. At low temperature, the spins of the Eu-atoms in the infinite cluster are all aligned, so that spin-disorder is weak. But if $x$ is small there is an infinite cluster of non-magnetic Ca-atoms, while Eu-clusters are tiny, on average, and sparse. Hence, spin-disorder is again weak. However, for $x$ in some range below $x_{\rm c}$, and in zero external magnetic field, there is considerable disorder in the way spins in different Eu-clusters are aligned. This enhances scattering of conduction electrons at different Eu-clusters, and one expects that the mobility edge, $E_*$, separating low-lying \textbf{localized} orbitals from \textbf{extended states} near the center of the conduction band is shifted away from the band edge towards the center of the band. If the Fermi energy in the conduction band is approximately constant as $x$ varies one is led to predict that Mott transitions may be observed at some $x^*\gtrsim x_{\rm c}$ and some $x_*\ll x_{\rm c}$; see Figure 5.
\begin{figure}[h]\center
\psfrag{E}{$E$}
\psfrag{Es}{$E_*(x)$}
\psfrag{EF}{$E_{\rm F}$}
\psfrag{0}{$0$}
\psfrag{1}{$1$}
\psfrag{x}{$x$}
\psfrag{xu}{$x_*$}
\psfrag{xo}{$x^*$}
\psfrag{xc}{$x_{\rm c}$}
\includegraphics[scale=0.5]{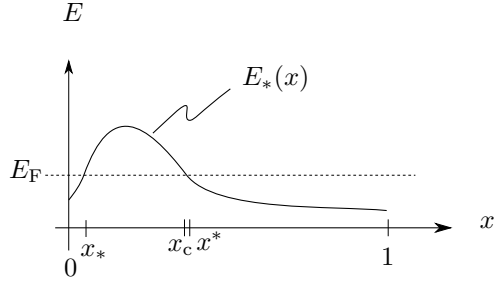}
\caption{Mobility edge, $E_*(x)$, as a function of $x$.}\label{}
\end{figure}

\section{Idealized model for Mott transitions driven by spin disorder}
In this section, we propose a model expected to exhibit some of the phenomena described in the last section, namely the Mott transition and the colossal magnetoresistance observed in $\mathrm{Eu}_x\mathrm{Ca}_{1-x}\mathrm{B}_6$ alloys. Our model is idealized to an extent that some of its properties, in particular Anderson localization, can be established rigorously. 

Because, experimentally, the conduction band of $\mathrm{Eu}_x\mathrm{Ca}_{1-x}\mathrm{B}_6$ is only weakly populated, $n_{\rm c}\lesssim O(10^{-3})$, we neglect interactions among conduction electrons and describe the propagation of a conduction electron with the help of a one-particle model. It is convenient to make use of a tight-binding approximation. The Hilbert space of pure state vectors of a conduction electron is then given by
\begin{align}\label{eq:fr1}
  \mathcal{H}=l^2(\ZZ^3)\otimes\CC^2\,.
\end{align}
Although valence electrons mediate an indirect exchange interaction between conduction electrons and the electrons in the half-filled $4f$ shells of Eu atoms, they do not appear explicitly in our model. Instead, the interactions of conduction electrons with the local Eu spins are described by a Heisenberg term coupling the spin of a conduction electron to the spin of a Eu atom localized in the same unit cell. Since the latter is quite large (s=7/2), we describe it by a \textbf{classical unit vector}, $\vec{m}$. The Heisenberg term then takes the form of a Zeeman term, $-J\vec{m}\cdot\vec{\sigma}$, where $\vec{\sigma}$ is the vector of Pauli matrices associated with a conduction electron, and $J>0$ is a constant. If a unit cell $j\in\ZZ^3$ is filled with a Eu atom then $|\vec{m}_j|=1$; if it is filled with a Ca-atom then $|\vec{m}_j|=0$. Eu-- and Ca-atoms are assumed to be distributed over the unit cells of $\ZZ^3$ by a \textbf{site percolation process}, with probability $x$ to place a Eu atom at any given site. The configuration, $\{\vec{m}_j\}_{j\in\ZZ^3}$, of classical spins is treated as quenched (in particular time-independent). Because of the observed tendency of Eu-spins in a connected Eu-cluster to order ferromagnetically, the distribution of the configurations $\{\vec{m}_j\}_{j\in\ZZ^3}$ of Eu-spins in every connected Eu-cluster, $\mathcal{C}$, is chosen to be given by a Gibbs measure
\begin{align}\label{eq:fr2}
\d\PP_{\mathcal{C}}(\vec{m}):=Z_{\mathcal{C}}^{-1}\exp\{\kappa\sum_{\substack{i,j\in\mathcal{C}\\|i-j|=1}}\vec{m}_i\cdot\vec{m_j}+\beta H\sum_jm_j^{(3)}\}\prod_{j\in\mathcal{C}}\delta(|\vec{m}_j|^2-1)\d^3m_j\,.
\end{align}
Here, $\kappa=\kappa(T)$ is a temperature-dependent, positive constant ($\kappa(T)$ is decreasing with $T$), $\beta$ is proportional to the inverse temperature, $H$ is the strength of a uniform external magnetic field in the $z$-direction, $m_j^{(3)}$is the $z$-component of $\vec{m}_j$, and $Z_{\mathcal{C}}$ (the cluster partition function) is chosen such that $\d\PP_{\mathcal{C}}$ is a probability measure.

The distribution of the Eu-clusters, $\mathcal{C}$, is given by an independent site percolation process with density $x$.

One might envisage to combine the distribution of the Eu-clusters $\mathcal{C}$ and of the configurations $\{\vec{m}_j\}_{j\in\ZZ^3}$ (with $\vec{m}_j=0$ if $j$ is occupied by a Ca atom) into a single probability distribution that would then describe a tendency towards Eu-Ca phase segregation. To simplify matters, we will not consider this possibility in the present paper.

The one-particle tight-binding Hamiltonian is chosen to be given by
\begin{align}\label{eq:fr3}
  h(\omega):=T+w_j(\omega)-J\vec{m}_j(\omega)\cdot\vec{\sigma}\,,
\end{align}
where $T\overset{e.g.}{=}-\Delta$ is a short-range hopping term ($\Delta$ is the discrete Laplacian), $\omega$ denotes the randomness of the interaction terms, $\{\vec{m}_j(\omega)\}$ is distributed according to (\ref{eq:fr2}), and $w(\omega)$ is a Bernoulli random potential with distribution
\begin{align}\label{eq:fr4}
w_j(\omega)=\begin{cases}w&\textnormal{if $\vec{m}_j\neq 0$ (that is, $j$ occupied by Eu)}\\-w&\textnormal{if $\vec{m}_j=0$ (that is, $j$ occupied by Ca)}\end{cases}\,.
\end{align}
The potential $w$ is incorporated in (\ref{eq:fr3}) because the potential energy of a conduction electron at a site $j$ may depend on whether $j$ is occupied by a Eu atom or a Ca atom.

The quantity of main interest to us, in this paper, is the electrical conductivity
\begin{align}\label{eq:fr5}
\sigma=\frac{e^2}{h}D\,,
\end{align}
where $D$ is the diffusion constant of conduction electrons. At temperature $T=0$ and for a given Fermi energy $E_{\rm F}$, $D$ is given by
\begin{align}\label{eq:fr6}
D=\int_{-\infty}^{E_{\rm F}}\d E\rho(E)D(E)\,,
\end{align}
where $\rho(E)$ is the density of states. In linear response theory, $D(E)$ is given by the Kubo formula
\begin{align}\label{kubo}
 \rho(E)D(E)&=\lim_{\eps\to 0}\frac{2\eps^2}{3\pi}\sum_{j\in\ZZ^3} |j|^2\, \mathbbmss{E} |\langle 0|(h(\omega)-E-\i\eps)^{-1}|j\rangle|^2\,,
\end{align}
where $\EE$ denotes an expectation with respect to the distributions given in (\ref{eq:fr2}) and (\ref{eq:fr4}), with Eu-clusters constructed according to an independent site percolation process of density $x$.

In order to keep our exposition simple, we consider some limiting regimes of the model introduced in (\ref{eq:fr2})--(\ref{eq:fr4}).
\\

(A) $\kappa\to 0$, $H$ small.

This regime is appropriate to describe elec\-tronic pro\-perties of\\ $\mathrm{Eu}_x\mathrm{Ca}_{1-x}\mathrm{B}_6$ in the absence of magnetic order (e.g.\, for $x$ very small, or well above the Curie temperature of the magnetic transition). Mathematically, this is the easiest regime. Using methods developed in \cite{bourgain09}, it is not very difficult to establish Anderson localization, for a fixed value of $x>0$, provided the energy lies sufficiently close to the band edges (depending on $x$ and the value of the constant $J$ in (\ref{eq:fr3})). Coexistence of localized states corresponding to energies in the band tails and extended states corresponding to energies near the center of the conduction band is expected for $J$ small enough. However, the nature of the spectrum of the random Schr\"odinger Hamiltonian $h(\omega)$ defined in (\ref{eq:fr3}) near the center of the energy band is very poorly understood, at present.
\\

(B) $\kappa\to\infty$, $H$ small.

In this regime, the spins of Eu atoms in every connected Eu-cluster are completely aligned, but their direction can vary arbitrarily from one such cluster to another one. Anderson localization can be proven for energies sufficiently close to the band edges, for sufficiently small values of the Eu concentration, $x$, so that, in particular, Eu does not percolate. For $x$ sufficiently close to 1 an infinite, ferromagnetically ordered Eu-cluster of density fairly close to $x$ exists, and we expect to find two mobility edges close to the band edges. (A mobility edge separates energies corresponding to localized states from energies corresponding to extended states.) Mathematically, the existence of mobility edges remains, however, an open issue.
\\

(C) $H\to\infty$.

In this limit, all the spins $\vec{m}_j$ are aligned in the positive $z$-direction. The conduction band then splits into two independent subbands for electrons with spin in the negative $z$-direction and those with spin in the positive $z$-direction, respectively. Within each subband, the Hamiltonian $h(\omega)$ is then equivalent to a ``Bernoulli Hamiltonian''
\begin{align}\label{eq:fr8}
h(\omega)=T+v(\omega)\,,
\end{align}
where
\begin{align}
v_j(\omega)\equiv v_j^{\pm}(\omega):=\begin{cases}w_j\pm J\,,& \vec{m}_j\neq 0\\-w_j\,,&\vec{m}_j=0\end{cases}\,.
\end{align}
Adapting methods developed in \cite{spencer93}, we show that, at all energies in small intervals attached to the band edges, except possibly in subsets of those intervals of very small Lebesgue measure, the quantity $\rho(E)D(E)$ introduced in (\ref{kubo}) vanishes, as long as $x\neq 0,1$. However, it is not known whether corresponding eigenstates are exponentially localized. It should be pointed out that the localization effect in this regime seems to be very weak, as can be seen in the inset of Figure 4.\\

Assuming that mobility edges, $E_*$, exist, separating energies $E$ with \\$\rho(E)D(E)=0$ from energies $E'$ closer to the center of the band, where $\rho(E')D(E')>0$, we expect (on the basis of our mathematical analysis of regimes (A) and (B)) that, as a function of $H$, $E_*=E_*(H)$ moves ever closer to a band edge, as $H$ increases (that is, as magnetic disorder decreases). Thus, for $0<x\lesssim 0.3$ and for a small, but positive density of conduction electrons, it can be expected that, at zero temperature, our model describes a Mott transition from an insulating state at small values of the magnetic field $H$ to a conducting state at large values of $H$. If correct this conjecture would explain the colossal (negative) magnetoresistance observed at $x=0.27$ and very low temperatures, recall Figure 4.

In the remaining sections of this paper, we switch gears, from physical reasoning to mathematics. We state some mathematical results proven for the model introduced in (\ref{eq:fr2})--(\ref{eq:fr4}), in regimes (A),(B), and (C), and we present outlines of proofs.

\section{Main mathematical results}
For the model of the electronic structure of $\mathrm{Eu}_x\mathrm{Ca}_{1-x}\mathrm{B}_6$ introduced in Section 3, we are able to prove the following main results.\\

In regime (A):

\begin{thm}The spectrum of $h(\omega)$ close to the band edges is almost surely pure point, and the eigenfunctions decay exponentially fast. A fortiori,  $\rho(E)D(E)$ vanishes for $E$ in the pure point spectrum.
\end{thm}

In regime (B):

\begin{thm}For $x<x_c$, the spectrum of $h(\omega)$ close to the band edges is almost surely pure point, and the eigenfunctions decay exponentially fast. A fortiori,  $\rho(E)D(E)$ vanishes for $E$ in the pure point spectrum.
\end{thm}

If the Fermi energy lies in the pure point spectrum, the alloy is an insulator. The region where we can prove pure point spectrum is shown to shrink with increasing $x$ (decreasing disorder), suggesting the conjecture that conduction increases as the absolutely continuous part of the spectrum approaches the Fermi energy.\\

In regime (C), we consider, as discussed, the Bernoulli Hamiltonian (for each subband)
\begin{align}\label{bernoulli}
h(\omega)=-\Delta+v_j(\omega)\,,
\end{align}
where we choose the energy scale such that 
\begin{align*}
v_j(\omega)=\begin{cases}v&\textrm{with prob $x$}\\-v&\textrm{with prob $1-x$}\,.\end{cases}
\end{align*}
\begin{thm}In the band tail from $[-v,-v^2]$, $\rho(E)D(E)$ vanishes for all energies outside a set of the order of $\exp(-\exp(1/\sqrt{v}))$.
\end{thm}
This is a perturbative result, and is meaningful only for small values of $v$.

\section{Outline of proofs}
Although rigorous proofs of the theorems stated in the previous section are somewhat beyond the scope of this article (see \cite{bourgain09,doktorarbeit} for details), we try to convey the main ideas underlying these proofs in the following.

The resolvent, or Green function, $G(E)=(h-E)^{-1}$, is closely related to the unitary time evolution, $\e^{-\i th}$, via the Laplace transform
\begin{align*}
G(z)\psi=\i \int_0^{\infty}\e^{\i tz}\e^{-\i th}\psi \d t,
\end{align*}
for $\mathrm{Im}\,z>0$. It is therefore not surprising that detailed knowledge of the resolvent is a means for investigating transport properties of the system described by the Hamiltonian $h$. We consider the random Hamiltonian $h(\omega)$ given in (\ref{eq:fr3}), but for clarity of presentation we will treat the two random terms separately, that is, we first set $w\equiv 0$.

 The Kubo formula
 \begin{align}
 \rho(E)D(E)&=\lim_{\eps\to 0}\frac{2\eps^2}{3\pi}\sum_{x\in\ZZ^3} |x|^2\, \mathbbmss{E} |\langle 0|(h(\omega)-E-\i\eps)^{-1}|x\rangle|^2
 \end{align}
shows that in order to establish absence of diffusion, we have to control
\begin{align*}
\EE |G(E+\i0;0,x)|^2\,,
\end{align*}
with
\begin{align*}
  G(z;0,x):=\langle 0|(h(\omega)-z)^{-1}|x\rangle\,.
\end{align*}
The way we do this is to prove that, for a suitable choice of the energy $E$, depending on disorder, $|G(E+\i0;0,x)|$ decays exponentially for large $|x|$, with probability approaching 1, as $|x|\to \infty$. This will enable us to deduce localization (that is, pure point spectrum of $h$, almost surely) in certain energy regimes.

It is technically convenient to study a regularized version of the full Green function: We restrict the Hamiltonian to finite-size cubes $\L\subset\ZZ^3$ (with appropriate boundary conditions), with the intention of sending $\L\nearrow\ZZ^3$ in the end. That is, we consider the matrix Hamiltonians
\begin{align*}
h_{\L_l}=-\Delta_{\L_l}+\11_{\L_l}(j)(w_j-J\vec{m}_j\cdot\vec{\sigma})\,,
\end{align*}
where $\Delta_{\L_l}$ is the finite-difference Laplacian restricted to $\L_l\subset\ZZ^3$, a cube of side length $l$, and $\11_{\L_l}$ is the characteristic function of $\L_l$. We then define
\begin{align*}
  G_{\L_l}(z):=(h_{\L_l}-z)^{-1}\,.
\end{align*}

The two main techniques used in our proofs are perturbation theory around ``good'' cubes, using the second resolvent equation, and induction on the scale of the cubes---a so-called ``multiscale analysis''. A ``good'' cube is defined as follows. \\
\textbf{Definition.} A cube $\L_l$ is called ``good'' at energy $E$ if the resolvent $G_{\L_l}(E)$ is well-behaved, that is, if there exists a constant $c>0$ such that
\begin{align}\label{wegner1}
&\norm{G_{\L_l}(E)}<\e^{l^{1/2}}\\
&|G_{\L_l}(E;x,y)|<\e^{-c|x-y|} \qquad \textrm{for }\; |x-y|\geq\frac{l}{10}\,.\label{offdiag}
\end{align}
If a cube is not ``good'', it is called ``bad''. Condition (\ref{wegner1})---a so-called Wegner-type estimate---ensures that $E$ does not lie too close to an eigenvalue of the matrix $h_{\L_l}$, and condition (\ref{offdiag}) represents the off-diagonal decay typical for the resolvent of a Schr\"odinger operator for $E$ in a spectral gap.\\
\subsection{Multiscale analysis}
The main idea of multiscale analysis is simple. The various scales are the sidelengths $l$ of the cubes $\L_l$. We establish existence of ``good'' cubes with high probability at an initial scale, and use this information to show existence of ``good'' cubes with even higher probability, inductively, at larger scales. The probability that a cube $\L_l$ is ``good'' should behave like $1-p_l$, where $p_l\propto l^{-k}$, and $k$ is some integer. The key point is that, at finite scales, we have to prove existence of ``good'' cubes only with high probability and not with probability one: Denote by $\Omega$ the whole probability space, that is, the set of all configurations $\omega$ of the random magnetic moments $\{\vec{m}_j\}_{j\in\ZZ^3}$. \textit{From the outset, we can discard configurations $\omega$ that are difficult to handle, as long as they have small probability.}\\
\textbf{Induction step.} The induction step from scale $l_n$ to $l_{n+1}\lesssim l_n^2$ proceeds as follows: Consider a cube $\L_{l_{n+1}}$ of size $l_{n+1}$. Since ``bad'' cubes of size $l_n$ are, by the induction hypothesis, improbable, the probability that there are \textit{many}, say $N$, ``bad'' $l_n$-cubes in a $l_{n+1}$-cube is even smaller, of the order of $p_{l_n}^N$. The integer $N$ is chosen such that $p_{l_n}^N<p_{l_{n+1}}$. \textit{Therefore, only configurations $\omega$ where the $l_{n+1}$-cube contains not too many ``bad'' $l_n$-cubes need to be considered.} 

Next, we describe the basic perturbation step. Consider a set $X\subset\ZZ^d$ and $Y\subset X$. The second resolvent identity yields
\begin{align*}
G_X=G_Y\oplus G_{X\backslash Y}+G_Y\oplus G_{X\backslash Y}\Gamma G_X\,,
\end{align*}
where $\Gamma$ is the boundary operator linking $Y$ to $X\backslash Y$. For example, for $x\in Y$ and $y\in X\backslash Y$, we have that
\begin{align*}
G_X(E;x,y)=\sum_{(z,z')\in\partial Y}G_Y(E;x,z)G_X(E;z',y)\,.
\end{align*}
To get an estimate on $|G_{\L_{l_{n+1}}}(E;x,y)|$ we proceed in the following way. In a first step, we excise the ``bad'' $l_n$-cubes (of which there are not too many) from $\L_{l_{n+1}}$. Iteration of the resolvent identity along a sequence of $l_n$-cubes the first of which is centered at $x$, the second on the boundary of the first, and so on until the last contains $y$, and the induction hypothesis (\ref{offdiag}) at scale $l_n$ establish the estimates (\ref{wegner1}) and (\ref{offdiag}) at scale $l_{n+1}$ for $\L_{l_{n+1}}$, with the ``bad'' cubes excised. The difficulty is now that, to couple a ``bad'' cube $\L_{l_n}$ back to $\L_{l_{n+1}}$, we need a slightly larger cube covering $\L_{l_n}$ that satisfies a Wegner-type estimate (\ref{wegner1}). But, by the induction hypothesis, the Wegner estimate holds for cubes of this size only with probability $1-p_{l_n}$, whereas we should like to establish it with the much larger probability $1-p_{l_{n+1}}$.

For random potentials with a \textit{bounded} probability density, it has been known for a long time \cite{wegner81} how to establish a Wegner-type estimate \eqref{wegner1} with probability $1-\e^{-l^{1/2}}$ \textit{simultaneously} on all scales $l$. After the inception of the rigorous study of Anderson localization \cite{froehlich83}, much effort has been devoted to the study of Wegner estimates, and there have been many advances, particularly in the continuum (see \cite{hislop02} and references therein). On the lattice however, for more general probability distributions, or, as in our case, matrix-valued random potentials, Wegner estimates are difficult to establish. Recent mathematical work \cite{bourgain09}, triggered by our study of the hexaboride alloys, shows how to establish a Wegner estimate inductively.\\
\textbf{Inductive Wegner estimate.} We have already argued that we can restrict our attention to configurations $\omega$ where there are not too many ``bad'' cubes of size $l_n$ in $\L_{l_{n+1}}$---call these ``bad'' cubes $B_k$. The key idea is to modify each configuration $\omega$ by changing in $\omega=\{\vec{m}_j\}_{j\in\L_{l_{n+1}}}$ only the values of $\vec{m}_j$ for $j$ \textit{inside} the $B_k$ such that, for this new configuration $\omega'$, each ``bad'' cube has a neighbourhood satisfying a Wegner estimate. In a second step, one shows with the help of complex analysis that these configurations $\omega'$ have actually very large probability.

Technically, this is done as follows (see also Figure \ref{figcube}):
\begin{figure}[h]\center
\psfrag{Q}{$B$}
\psfrag{Qi}{$B_{(i)}$}
\psfrag{n}{$\L_{l_n}$}
\psfrag{n+1}{$\L_{l_{n+1}}$}
\includegraphics[scale=0.5]{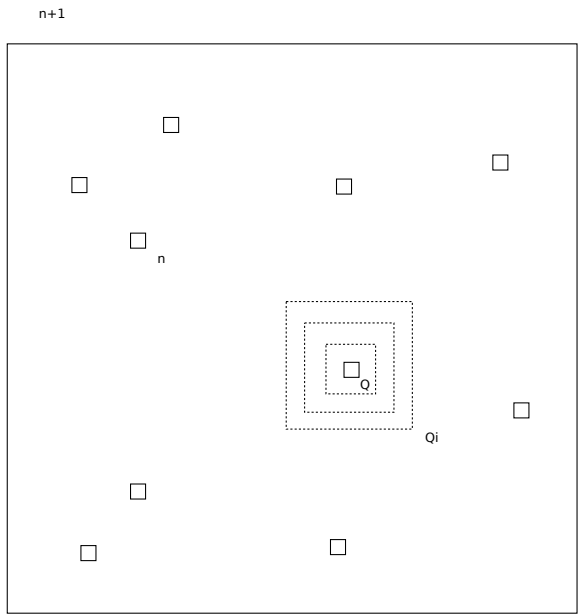}
\caption{Generic cube of size $l_{n+1}$ with not too many ``bad'' cubes of size $l_n$.}\label{figcube}
\end{figure}
Pick one of the ``bad'' cubes, call it $B$. By the induction hypothesis, we know that a neighbourhood of a ``bad'' cube $B$ satisfies a Wegner estimate with probability $1-p_{l_n}$ (the fact that such a neighbourhood is actually slightly larger than $l_n$ is a technical nicety and of minor importance). Now, cover the ``bad'' cube $B$ with a slightly larger cube $B_{(1)}$. Then the probability that the configuration in the ring $B_{(1)}\backslash B$ makes the cube $B_{(1)}$ ``bad'', no matter what the configuration inside $B$, is smaller than $p_{l_n}$. Thus, the probability that there are \textit{many}, say $N$, equicentered cubes $B_{(i)}$ of increasing diameter such that, for all $i$, the configuration in the ring $B_{(i)}\backslash B_{(i-1)}$ makes the cube $B_{(i)}$ ``bad'', no matter what the configuration in the interior, $B_{(i-1)}$, is very small, namely of the order of $p_{l_n}^N$. We can therefore, from the outset, restrict considerations to configurations $\omega$ where each bad cube $B$ can be replaced by a larger cube $B_{(i_0-1)}$ with the property that the configuration $\omega$ can be modified \textit{inside the cube $B_{(i_0-1)}$ alone} such that the cube $B_{(i_0-1)}$ has a ``good'' neighborhood $B_{(i_0)}$. The modified configuration shall be denoted by $\omega'$.

We can now use again the resolvent identity to establish the bound \eqref{wegner1} on $\norm{G_{\L_{l_{n+1}}}(E,\omega')}$. We have thus found that $\norm{G_{\L_{l_{n+1}}}(E)}$ is bounded at scale $l_{n+1}$ for \textit{one fixed configuration $\omega'$}. One may think that this is far too little, since a single point $\omega'$ has zero probability in $\Omega$. However, in a last step, one shows, using a matrix-valued Cartan-type lemma, that the probability of configurations satisfying the estimate (\ref{wegner1}) at scale $l_{n+1}$ is very large.

To understand the last step we need to recall a result from complex analysis, known as Cartan's lemma. The precise mathematical statement can be found in the appendix. For our purposes, a rough understanding of the lemma will suffice: The lemma says that an analytic function that is bounded away from zero at \textit{one} point of its domain of definition is at most points not too close to zero. We need a higher-dimensional matrix-valued analog of this lemma due to Bourgain \cite{bourgain09}, which we relegate to the appendix. This generalized lemma says that if $\norm{G_{\L_{n+1}}(E;\omega)}$ is not too large for \textit{one} configuration $\omega'$, it is not too large for most configurations; actually only exponentially few (in $l_n$) configurations have to be excluded. \textit{It is important to point out that in order to apply this lemma the distribution of magnetic moments on the unit sphere needs to have a bounded density with respect to the uniform measure on the sphere.}

We can now apply the lemma, for we have explicitly constructed a configuration $\omega'$ with the desired properties.

Using \eqref{wegner1}, the off-diagonal decay \eqref{offdiag} is easily established as before by expanding the resolvent along a sequence of nested cubes.

The induction step is thus completed, because we have verified that \eqref{wegner1} and \eqref{offdiag} hold with probability at least $1-p_{l_{n+1}}$ at scale $l_{n+1}$.

\textbf{Induction basis.} To complete the proof we need to verify conditions \eqref{wegner1} and \eqref{offdiag} at some initial scale $l_0$. The proof is based on the following intuition. If a configuration $\omega$ is to be such that $G_{\L_{l_0}}(E,\omega)$ has an eigenvalue close to the upper edge of the spectrum $E_0=4d+|m|$, then, first, most of the sites in $\L_{l_0}$ have to be occupied by a magnetic moment, and, second, most of nearest-neighbor moments have to be closely aligned, to maximize the hopping term. But since the moments are independent, a large-deviations estimate shows that such a configuration occurs with very small probability. (Here we use the fact that the probability of large clusters is exponentially small in the cluster size, for $x<x_c$.)\\
\textbf{Absence of diffusion.} From the exponential decay of $G_{\L_l}(E)$, with probability $\simeq 1-\textrm{const}\,l^{-k}$, we get \textit{absence of diffusion}, using that
\begin{align*}
\lim_{\L\nearrow\ZZ^3}G_{\L}(E;x,y)=G(E;x,y)\,,
\end{align*}
and the Kubo formula (\ref{kubo}) for the diffusion constant $D$.
\\
\textbf{Pure point spectrum.} Call $I$ the set of energies where we can prove almost sure exponential decay of $G(E;x,y)$ (that is, $I$ lies in the band tails). To prove pure point spectrum in $I$, we need additional ideas, see \cite{froehlich85,bourgain05}. The first is that, with respect to the spectral measure of $h$, almost every energy is a \textit{generalized} eigenvalue, that is, a polynomially bounded solution to the equation
\begin{align*}
(h-E)\psi=0\,.
\end{align*}
Next, we note that we can express such a solution by the Green function,
\begin{align*}
\psi(x)=\sum_{(z,z')\in\partial R_n}G_{R_n}(E;x,z)\psi(z') \qquad \forall x\in R_n\,,
\end{align*}
where $R_n$ is most conveniently chosen a ring, say $\Lambda_{l_{5n}}\backslash\L_{l_n}$. The exponential decay of the Green function in $R_n$ implies thus exponential decay of $\psi(x)$ for $x$ large enough, which, in turn, means that $\psi$ is a true eigenfunction, and, a fortiori, that $E$ is an eigenvalue.

\subsection{Bernoulli random variables}
If the random potential is of Bernoulli-type, the above inductive method for proving a Wegner estimate breaks down, since the ``density'' of the distribution is in this case unbounded. However, there is an old method \cite{spencer93} (see also \cite{elgart09}) that gives partial results. We should also mention that the Bernoulli problem has been solved in the continuum \cite{bourgain05}, but with methods that do not extend to lattice operators.\\

In order to present the main ideas of our analysis, we introduce the function $N_\L(E,\omega)$ which counts the number of eigenvalues of $h_\L(\omega)$ less than $E$ (recall that the Hamiltonian studied in this section is given by (\ref{bernoulli})). In the thermodynamic limit, an ergodic theorem guarantees the existence of the so-called \textit{integrated density of states}
\begin{align*}
N(E)=\lim_{\L\nearrow\ZZ^3}\frac{N_\L(E,\omega)}{|\L|}\,,
\end{align*}
which is independent of $\omega$ almost surely. Its derivative in $E$ is the density of states and is denoted by $\rho(E)$. We can express $N_\L(E,\omega)$ by the resolvent $G_\L(E,\omega)$ by noting that
\begin{align*}
N_\L(E,\omega)=\tr \mathrm{P}_\L(E,\omega)=\sum_{x\in\L}\scalar{\delta_x}{\mathrm{P}_\L(E,\omega)\delta_x}\,,
\end{align*}
where $\mathrm{P}_\L(E,\omega)$ denotes the spectral projection of $h_\L(\omega)$ onto the interval $(-\infty,E)$, and
 \begin{align*}
\mathrm{P}_{\varLambda}(E,\omega)=\underset{\eps\rightarrow 0}{\mathrm{s-}\lim}\frac{1}{\pi} \int_{-\infty}^E \mathrm{Im}\,G_{\varLambda}(E'+\i\eps,\omega)\d E'\,.
\end{align*}
Using that
\begin{align*}
\frac{2}{\pi} \int_{-\infty}^E \d E' \frac{\eps}{(x-E')^2+\eps^2} = \frac{2}{\pi} \arctan(\frac{E-x}{\eps})+1\geq \Theta(E-x)\,,\quad \forall \eps>0\,,
\end{align*}
we see that the spectral theorem provides the bound
\begin{align}\label{boundondensity}
N(E)=\mathbbmss{E}\, \scalar{\delta_{x_0}}{\mathrm{P}(E,\omega)\delta_{x_0}}\leq \frac{2}{\pi}\int_{-\infty}^E\d E'\, \mathrm{Im}\, \mathbbmss{E}\, G(E'+\i\eps;x_0,x_0)\,,\quad \forall \eps>0\,,
\end{align}
the first equality following from translation invariance. We can bound the probability that an eigenvalue of $h_\L$ exists below $E$ in terms of $N(E)$,
\begin{align*}
\mathbbmss{P}[N_{\varLambda}(E)\geq 1]=\mathbbmss{E} \mathbbmss{1}_{\{N_{\varLambda}(E)\geq 1 \}}\leq \mathbbmss{E} N_{\varLambda}(E) \mathbbmss{1}_{\{N_{\varLambda}(E)\geq 1 \}} \leq \mathbbmss{E}\,N_{\varLambda}(E) \,,
\end{align*}
and since Dirichlet boundary conditions raise eigenvalues we have that
\begin{align*}
\mathbbmss{E}\,N_{\varLambda}(E) \leq |\L|N(E)\,.
\end{align*}
The main idea is now to show that the integrated density of states in the band tails is so small that with very high probability, $h_\L$ does not have an eigenvalue below $E$. Then we can invoke the following well-known lemma to prove exponential decay of the off-diagonal elements.
\begin{lem*}[Combes-Thomas]
Whenever $h_{\L}$ has no spectrum below $E_1$ then for $E<E_1$,
\begin{align}\label{simon}
|G_{\varLambda}(E;x,y)|\leq \frac{2}{\delta} \mathrm{exp}(-\textnormal{const}\sqrt{\delta} |x-y|),
\end{align}
where $\delta:=|E-E_1|$. 
\end{lem*}
Recalling equation (\ref{boundondensity}), we see that, in order to get a bound on the integrated density of states, we have to bound the imaginary part of the averaged Green function. The strategy is simple. If $v$ is small the first thought is to expand the resolvent $(-\Delta+v(\omega)-E-\i\eps)^{-1}$ in powers of $v$ around $(-\Delta-E)^{-1}$. But since we consider the \textit{average} $\EE (-\Delta-v(\omega)-E)^{-1}$, there may be an optimal energy $E_0$ around which to expand. This energy is found as follows. We introduce a perturbation parameter $\l$, writing
\begin{align*}
h=-\Delta+\l v\,,
\end{align*}
with $v=\{v_j\}_{j\in\ZZ^3}$ a collection of independent Bernoulli random variables, $v_j=\pm 1$ with probability $\frac{1}{2}$, each; (for simplicity, we consider the symmetric case, $x=\frac{1}{2}$; but the general case is hardly more difficult).  Formally, we have that
\begin{align*}
G(E+\i\eps)&=\frac{1}{-\Delta+\l v-E-\i\eps}=\frac{1}{-\Delta+E_0-\i\eps+\l v+(-E_0-E)}\\
&=G_0\sum_{n\geq 0}((-\l v+E+E_0)G_0)^n\,,
\end{align*}
where we have introduced the unperturbed Green function
\begin{align*}
G_0(-E_0+\i \eps) = (-\Delta + E_0-\i \eps)^{-1} \,.
\end{align*}
The first terms are
\begin{align*}
G&=G_0-G_0[\l v-(E+E_0)]G_0+\\
&+G_0[\l v-(E+E_0)]G_0[\l v-(E+E_0)]G_0+\dots
\end{align*}
Using that $\EE v=0$ and $\EE v^2=1$, we obtain
\begin{align*}
\EE G=G_0+G_0^2[E+E_0]+\l^2 G_0(0,0) G_0^2+G_0^3[E+E_0]^2+\dots\,.
\end{align*}
Thus, we see that, in order for the $\l^2$-term to vanish, we must choose
\begin{align}\label{motivation}
E+E_0=-\l^2G_0(-E_0+\i\eps;0,0)=\mathcal{O}(\l^2)
\end{align}
to arrive at $\EE G=G_0+\mathcal{O}(\l^4)$.\\

To make the above considerations mathematically respectable, we iterate the second resolvent identity 
\begin{align*}
G(E+\i \eps;x_0,x) &= G_0(-E_0+\i \eps;x_0,x)+ \\
&\sum_y G_0(-E_0+\i \eps;x_0,y) [-\lambda v +(E+E_0)] G(E+\i \eps;y,x) \,,
\end{align*}
with $0<E_0=-E-\lambda^2 \Sigma(E)$, where $\lambda^2 \Sigma(E)$ is motivated by (\ref{motivation}) and is defined in a self-consistent way by
\begin{align*}
\lambda^2 \Sigma(E) = \lambda^2(-\Delta -E-\lambda^2 \Sigma(E) -\i \eps)^{-1}(0,0) \,.
\end{align*}
Since $E_0$ has to be positive in order for $G_0(E_0)$ to exist we have as an upper limit for $E$
\begin{align}\label{Estar}
E^*=\l^2 \Delta^{-1}(0,0)<0\,.
\end{align}
Next, we iterate the resolvent identity $M$ times with the intention of optimizing the truncation parameter $M$, later on, to minimize the remainder term. Setting $W:=-\lambda v -\lambda^2 \Sigma$ and $G^\eps:=G(E+\i\eps)$ we get
\begin{align*}
G^\eps = \sum_{m=0}^M G^\eps_0 \left[W G^\eps_0\right]^m + G^\eps_0 \left[W G^\eps_0 \right]^M W G^\eps \,.
\end{align*}
The dangling factor $G^\eps$ in the remainder is estimated trivially by $1/\eps$. The key observation is that the imaginary part of the first $M+1$ terms on the right-hand side can be shown to be proportional to $\eps$, whereas we will indicate below how to prove that the remainder term multiplying $G^\eps$ is of order $\e^{-(2\l)^{-1/2}}$. Thus we choose $\eps^2=\e^{-(2\l)^{-1/2}}$ and get the estimate
\begin{align*}
\mathrm{Im}\, \mathbbmss{E}\, G(E'+\i\eps;x_0,x_0)\leq \e^{-\frac{1}{2}(2\l)^{-1/2}}\,.
\end{align*}
In computing the terms of the form
\begin{align*}
&\mathbbmss{E}[G_0 (\l v G_0)^m](x_0,x)\\
=&\mathbbmss{E}[\sum_{x_1,\dots,x_m}G_0(x_0,x_1)\l v_{x_1} G_0(x_1,x_2)\dots \l v_{x_m} G_0(x_m,x)]\,,
\end{align*}
 it is easiest to use a graphical representation: $G_0(x,y)$ corresponds to a line joining $x$ and $y$, while the interaction $v$ corresponds to a vertex. Averaging over the randomness yields terms represented by graphs obtained by fusing an even number (since expressions involving an odd number of $v$ vanish upon averaging) of vertices at a time until none remains unpaired. Because 
$$\begin{array}{c}
1=\mathbbmss{E} v^{2n}\ll \gamma(n) \(\mathbbmss{E} v^2\)^n=\gamma(n)\\
\gamma(n)=\frac{(2n)!}{2^n n!}=\textnormal{number of full pair contractions}
\end{array}$$
we can use Wick's theorem to get an upper bound by considering only fusions of {\it pairs} of vertices. It is well known how to bound such graphs: Consider the bubble graph:
\begin{align*}
\lambda^2 \sum_x |G_0(0,x)|^2 &\leq C\lambda^2 \int \d^3 x\frac{\e^{-2\sqrt{E_0}|x|}}{(1+|x|)^2}\\
&=\lambda^2 \frac{C}{\sqrt{E_0}}\int\d^3x \frac{\e^{-2|x|}}{(\sqrt{E_0}+|x|)^2}\leq C \left(\frac{\lambda^2}{\sqrt{E_0}}\right) =:A\,.
\end{align*}
Strings of $n$ such bubbles are of size $A^n$ and add up to $A/(1-A)$ (convergence only holds if $E_0>\lambda^{4-2\delta}$).\\
It turns out that a graph of size $M$ is of the order of
\begin{align*}
\(\frac{\l^2}{\sqrt{E_0}}\)^M\,, \textnormal{ up to logarithmic corrections.}
\end{align*}
There are less than $\frac{(2M)!}{2^MM!}\simeq 2^M(M/\e)^M\sqrt{2}<2^MM!$ graphs generated by contracting pairs of $2M$ vertices, so the sum of all graphs is bounded by
\begin{align*}
2^MM! \(\frac{\l^2}{\sqrt{E_0}}\)^M\leq M!(2\l)^{M\delta}\,.
\end{align*}
We now optimize our choice of the truncation parameter $M$: Choosing $M=(2\l)^{-\delta}$
we get the bound
\begin{align*}
M!M^{-M}\leq C \e^{-M}=C\e^{-(2\l)^{-\delta}}\,.
\end{align*}
By similar calculations, collecting our previous estimates and setting $\delta=1/2$, we obtain the following upper bound for the integrated density of states,
\begin{align*}
N(E)\leq C\e^{-\frac{1}{2}(2\l)^{-1/2}}\,.
\end{align*}
We have thus shown that with probability larger than
\begin{align}\label{one}
1-|\L|C\e^{-\frac{1}{2}(2\l)^{-1/2}}
\end{align}
there is off-diagonal exponential decay of the resolvent $G_\L(E)$, for $E\in [-\l,-\l^2]$, see (\ref{Estar}). Therefore, for any initial scale $l_0$ there is a $\l$ such that (\ref{offdiag}) holds with high probability. Having established condition (\ref{offdiag}) of the multiscale analysis, we are looking for a Wegner estimate. Because of the very singular nature of the Bernoulli potential, the inductive scheme devised in \cite{bourgain09} does not work. In \cite{spencer93}, the following trick was introduced: For each scale $l_n$, we define $\mu(J)=\mu_{l_n}(J)$ to be the expected number of eigenvalues of $h_{\L_{l_n}}$ in an interval $J$. Next, we introduce the classical form of the Wegner estimate, and show that it implies (\ref{wegner1}). Since
\begin{align*}
\mathbbmss{P}[\norm{G_{\L_l}(E)}\geq \e^{l^{1/2}}]=\mathbbmss{P}[\mathrm{dist}(E,\sigma(h_{\varLambda_l}))\leq \e^{-l^{1/2}}],
\end{align*}
it is clear that (\ref{wegner1}) holds with probability $p_l=1-l^{-k}$ if we can show that
\begin{align*}
\mathbbmss{P}[\mathrm{dist}(E,\sigma(h_{\varLambda_l}))\leq \kappa]\leq C|\L_l|\kappa^{1/2}\,.
\end{align*}
The following easy estimate shows how we have to proceed:
\begin{align*}
\mathbbmss{P}[\mathrm{dist}(E,\sigma(h_{\varLambda_l}))\leq \kappa] \leq \mathbbmss{E} N_{\varLambda_l}(E+\kappa) -\mathbbmss{E} N_{\varLambda_l}(E-\kappa)=\mu_l(E-\kappa,E+\kappa)\,.
\end{align*}
We see that (\ref{wegner1}) is fulfilled if we exclude a set of ``singular'' energies, and the following lemma shows that this set of energies has very small measure.
\begin{lem*} Let $\mu$ be a measure on an interval $I$. Let $S$ be the set of energies $E$ for which the measure is singular at scale $\eps$, that is at which
 \begin{align*}
 \mu(E-\eps,E+\eps)\geq \eps^{1/2}\,.
 \end{align*}
If $|S|$ denotes the Lebesgue measure of $S$ then
\begin{align*}
|S|\leq 2 \mu(I) \eps^{1/2}\,.
\end{align*}
\end{lem*}
\begin{proof}
An easy application of Fubini's theorem shows that
\begin{align*}
\eps^{1/2} |S| &\leq \int_S \d E \int_I \d\mu(x) 1_{[E-\eps,E+\eps]}(x)\\
&\leq \int_I\d\mu \int_I \d E 1_{[E-\eps,E+\eps]}(x) =\mu(I) 2 \eps\,.
\end{align*}
\end{proof}
Appealing to multiscale analysis, we see that, at each scale $l_n\simeq l_0^{2^n}$, we have to exclude energies of measure $C \exp(-\frac{1}{2} l_0^{\frac{1}{2}2^n})$. Thus the total measure of energies we might have to excise is
\begin{align*}
|E^{\mathrm{exc}}|=|\bigcup_n E^{\mathrm{exc}}_n|\leq \sum_{n=0}^{\infty}C \exp(-\frac{1}{2} l_0^{\frac{1}{2}2^n}) \leq 2C\exp(-\frac{1}{2} l_0^{\frac{1}{2}}) \,.
\end{align*}
Because the integrated density of states is so small, we can choose $l_0$ to be exponentially large in $\lambda^{-1/2}$ (but not larger because of the factor $\L_l\propto l^3$ in (\ref{one})), and hence the set of energies we have to excise is of order $\exp(-\frac{1}{2}\exp(\frac{1}{2}\l^{-1/2}))$.\\

\appendix
\section{Cartan's Lemma}

\begin{lem*}\label{cartanestimate}Let $f(z)$ be a function analytic in the disc $\{z:|z|\leq \e R\}$, $|f(0)|=1$, and let $\eta$ be an arbitrary small positive number. Then we have that
\begin{align*}
\l(\{x\in [-R,R]:|f(x)|< \delta\}) \leq 30\e^3R\delta^\frac{1}{\log M_f(R)}\,,
\end{align*}
where $M_f(r)=\max_{|z|=r}|f(z)|$ and $\l(\cdot)$ denotes Lebesgue measure.
\end{lem*}
This lemma says that a function analytic in a disc that is bounded away from zero at \textit{one} point is not too close to zero at ``most'' points.\\
Since $G_\L(E)$ is a matrix, we need a generalization to matrix-valued functions, which is due to Bourgain \cite{bourgain09}.
\begin{lem*}\label{bourgainlemma2} Let $\Gamma(x)$ be a real analytic self-adjoint $N\times N$-matrix function of $x\in \Omega=[a,b]^n$, satisfying the following conditions (with $m\ll N$, $B_1,B_2,B_3>1$)
\newcounter{hola}\setcounter{hola}{1}
\begin{itemize} \item[(\arabic{hola})] $\Gamma(x)$ has an analytic extension $\Gamma(z)$ to $z\in\mathrm{D}^n:=\{z\in\CC: |z-\frac{a+b}{2}|<\frac{\e}{2}|a-b|\}^n$ with 
\begin{align*}\
\norm{\Gamma(z)}<B_1 \qquad z\in\mathrm{D}^n
\end{align*}
\addtocounter{hola}{1}\item[(\arabic{hola})] There is a subset $J$ of $\{1,2,\dots, N\}$ such that $|J|\leq m$ and for all $z\in \mathrm{D}^n$
\begin{align*}
\norm{\(R_{\{1,\dots,N\}\backslash J}\Gamma(z)R_{\{1,\dots,N\}\backslash J} \)^{-1}}<B_2\,,
\end{align*}
where $R_S$ denotes coordinate restriction to $S$
\addtocounter{hola}{1}\item[(\arabic{hola})] For some $\omega'\in \Omega$ we have
\begin{align*}
\norm{\Gamma(\omega')^{-1}}<B_3
\end{align*}
\end{itemize}
Then
\begin{align*}
\mu(\{x\in \Omega:\norm{\Gamma(x)^{-1}}>K \})< C\norm{g}_\infty n|a-b|^n K^{-\frac{c}{m\log B_1B_2B_3}}\,,
\end{align*}
\end{lem*}
where $\mu$ denotes a measure with bounded density $g$ with respect to Lebesgue measure, and $C,c>0$ are constants.

Assumption (2) is clearly fulfilled in our case of $\Gamma=h-E$, the subset $J$ being the lattice sites of the union of the bad cubes, and assumptions (1) and (3) are self-explanatory.
\bibliography{article}
\end{document}